\documentclass[preprint,xcolor=pst,dvips,11pt]{elsarticle}
\usepackage{vaucanson-g}
\biboptions{sort&compress}

\usepackage{tikz}

\input{Enfichier.sty}

\journal{Journal of Discrete Algorithms}

\begin{document}
\begin{frontmatter}

\title{Extended to Multi-Tilde-Bar Regular Expressions and Efficient Finite Automata Constructions}
 \author[label1]{Jean-Marc Champarnaud}
  \author[label]{Faissal Ouardi\tnoteref{label2}}
   \author[label1]{Djelloul Ziadi }
 \address[label1]{LITIS, University of Rouen, 76821 Saint-Etienne-du-Rouvray, France}
  \address[label]{Department of Computer Science, Faculty of Science, Mohammed V-Agdal University, Morocco}
\tnotetext[label2]{Corresponding author. E-mail address: \texttt{ouardi@fsr.ac.ma}}





\begin{abstract}
Several algorithms have been designed to convert a regular expression into an equivalent finite automaton.
One of the most popular constructions, due to Glushkov
and to McNaughton and Yamada,
is based on the computation of the $\nul$, $\First$, $\Last$ and $\Follow$ sets (called Glushkov functions) associated with a 
linearized version of the expression.
Recently Mignot
considered a family of extended expressions called Extended to multi-tilde-bar Regular Expressions (EmtbREs)
and he showed that, under some restrictions, Glushkov functions  can be defined for an EmtbRE.
In this paper we present an algorithm which efficiently computes the Glushkov functions of an unrestricted EmtbRE. 
Our approach is based on a recursive definition of the language associated with an EmtbRE
which enlightens the fact that the worst case time complexity of the conversion of an EmtbRE into an automaton
is related to the worst case time complexity of the computation of the $\nul$ function. 
Finally
we show how to extend the \ZPC-structure
to EmtbREs, which allows us to apply to this family of extended expressions the efficient constructions based on this structure
(in particular the construction of the c-continuation automaton, the position automaton, the follow automaton and the equation automaton).

\end{abstract}

\begin{keyword}
Regular Expressions and languages
 \sep Finite automata \sep
Computation Complexity 


\end{keyword}
\end{frontmatter}
\section{Introduction}
According to Kleene's theorem \cite{Kln}, regular expressions and finite automata are two equivalent representations of regular languages.
The conversion from a representation into the other one raised numerous research works.
Concerning the conversion of a regular expression into a finite automaton we can cite the following references:
\cite{Ant,BeSe,BePi,Brz,Brg,CNZ,CZ,CPZ,Glush,IlYu,MnYa,ZPC},
for which a common aim is to reduce the space and/or worst case time complexity of the result of the conversion.  
In this paper we are particularly interested by the implementation of conversion algorithms which are based on the notion of position,
such as the five first ones in the above list.
Following \cite{Glush,MnYa}, these algorithms are based on the the computation of the Null, First, Last and Follow sets (called Glushkov functions)
associated with a linearized version of the expression.
Recently Mignot \cite{MiT}
considered a family of extended expressions called Extended to multi-tilde-bar Regular Expressions (EmtbREs)
and he showed that, under some restrictions, the Glushkov functions can be defined for an EmtbRE (see also \cite{CCMacta,CCMtcs}).
In this paper we present an algorithm which efficiently computes the Glushkov functions of an unrestricted EmtbRE. 
Our approach is based on a recursive definition of the language associated with an EmtbRE
which enlightens the fact that worst case time complexity of the conversion of an EmtbRE into an automaton
is related to the worst case time complexity of the computation of the Null function. 
Finally
we show how to extend the \ZPC-structure \cite{ZPC}
to EmtbREs, which allows us to apply to this family of extended expressions the efficient constructions based on this structure 
(in particular the construction of the c-continuation automaton \cite{CZ}, the position automaton \cite{ZPC},
the follow automaton \cite{CNZ} and the equation automaton \cite{CZ,KOZ}).

The structure of the paper is as follows.
In Section 2, we recall some basic definitions concerning regular expressions and finite automata,
and we recall the notion of multi-tilde-bar expression.
New properties concerning the language of a multi-tilde-bar expression are stated in Section 3.
In Section 4, we give the definition of the position automaton associated with an arbitrary multi-tilde-bar expression.
Section 5 is devoted to an efficient computation of the position automaton of an EmtbRE, through the extension of the notion of
\ZPC-structure of a regular expression.
 
\section{Preliminaries}

\subsection{Regular expressions and finite automata}
Let $A$ be a non-empty finite set of symbols,
called an {\it alphabet}. The set of all the words over $A$ is denoted by $A^*$. The empty
word is denoted by $\varepsilon$. A {\it language} over $A$ is a subset of $A^*$. {\it Regular expressions}
over an alphabet $A$ and regular languages that they denote are inductively defined
as follows:
\begin{itemize}
\item $\emptyset$ is a regular expression denoting the language $\LL(\emptyset) = \emptyset$. 
\item $x$, for all $x\in A\cup \{\varepsilon\}$, is a regular expression 
denoting the language $\LL(x) = \{x\}$. 
\item Let $\f$ (resp. $\g$) be a regular expression denoting the language $\LL(\f )$ (resp.
$\LL(\g)$); then we have:
\begin{itemize}
\item {$\footnotesize{(\f +\g)}$} is a regular expression denoting the language  \\ {$\footnotesize{\LL(\f+\g) = \LL(\f)\cup\LL(\g)}$}.
\item $(\f\cdot \g)$ is a regular expression denoting the language \\ $\LL(\f\cdot \g) = \LL(\f)\cdot\LL(\g)$.
\item $(\f^*)$ is a regular expression denoting the language $\LL(\f^*) = (\LL(\f))^*$ .
\end{itemize}

\end{itemize}
The following identities are classically used: \\$\emptyset + \E = \E = \E + \emptyset,~ \varepsilon\cdot \E = \E = \E\cdot \varepsilon,
\emptyset\cdot \E = \emptyset = \E\cdot \emptyset$.\\
Let $\E$ be a regular expression. Its  {\it linearized form}, denoted by $\E'$, is obtained by
ranking every letter occurrence with a subindex denoting its position in $\E$. We say that a regular
 expression is in linear form if each letter of the expression occurs only once. 
Subscripted letters are called positions and the set of positions is denoted by $\pos(\E)$. We denote by $h$ the application that maps
each position in $\pos(\E)$ to the symbol of $A$ that appears at this position in $\E$. 
 The {\it size} of $\E$, denoted by $|\E|$, is the 
size of its syntactical tree. We call {\it alphabetic width} of $\E$, denoted by $||\E||$, the number of occurrences of 
letters in the expression.
\begin{df}\label{df2}
 Let $\E$ be a regular expression denoting the language $\LL$. The set $\nul(\E)$
 is defined by:
 $$\nul(\E)=\left\{\begin{array}{ll}
                                           \{\varepsilon\} & \mbox{ if }\varepsilon\in \LL,\\
                                           \emptyset & \mbox{ otherwise. }
                                          \end{array}\right.
                                                            $$ 
 
 \end{df}
A {\it finite automaton} (NFA) is a 5-tuple ${\cal A} = \langle Q, A, \delta, q_0 , F \rangle,$ where $Q$ is
a finite set of states, $A$ is an alphabet, $q_0 \in Q$ is the initial state, $F\subseteq Q$ is
the set of final states and $\delta : Q \times A  \longmapsto 2^Q$ is the transition function.
The language recognized by ${\cal A}$ is denoted by
$\LL({\cal A})$.

\subsection{Multi-tilde-bar expressions}
We now recall the syntactical definition of extended to multi-tilde-bar regular expressions (EmtbREs) \cite{CCMacta}.
Notice that these expressions will be proven to be regular later (see Corollary \ref{regular}).\smallskip

Let $\E$ be a regular expression.
The language $\LL(\E)\setminus \{\varepsilon\}$
is denoted by the expression $\overline{\E}$ (bar operator)
and the language $\LL(\E)\cup \{\varepsilon\}$
is denoted by the expression 
\begin{pspicture}(0,0)(0.4,0.35)
\pszigzag[coilwidth=0.07,coilheight=2.2,coilarm=0](0,0.32)(0.4,0.32)
\uput{0.1cm}[1](0,0.1){$\E$}
\end{pspicture} (tilde operator).
Without loss of generality, any regular expression can be considered as a product of concatenation $\E_1\cdot \E_2 \cdots \E_n$ of $n$ subexpressions , with $n\geq 1$.  
Such a product is denoted by $\E_{1,n}$ and the set of its factors is denoted by $F$.
Let us consider the set of pairs 
$F'=\{(i,j)~|~1\leq i\leq j\leq n\}$.
For $1\leq i\leq j\leq n$, the factor $\E_{i,j}$ is represented by the pair $(i,j) \in F'$.
A bar operator (resp. a tilde operator) applying on the factor $\E_{i,j}$ is also represented by the pair $(i,j) \in F'$.
Given two disjoint subsets $F_1$ end $F_2$ of $F$,
a {\it multi-tilde-bar operator} is defined by two subsets of $F'$: the set $\BB_1^n$ of bar operators applying on the factors of $F_1$ and the set $\TT_1^n$ of tilde operators applying on the factors of $F_2$.
Finally, a {\it multi-tilde-bar expression} $\E'_{1,n}$ is defined as a product $\E_{1,n}$ equipped with a set $\BB_1^n$ of bars and a set $\TT_1^n$ of tildes.
 
 \begin{df}{\cite{CCMacta}}
 An {\it Extended to multi-tilde-bar Regular Expression} (EmtbRE) over an
alphabet $A$ is inductively defined by:\smallskip

$\begin{array}{lll}
\E = \emptyset, &~~~ & \E = (\f+\g), \mbox{ with } \f \mbox{ and } \g \mbox{ two EmtbREs},\\
\E = x, \mbox{ with } x \in A\cup\{\varepsilon\},  &~~~&\E = (\f\cdot\g), \mbox{ with } \f \mbox{ and } \g \mbox{ two EmtbREs},\\
 &~~~& \E = (\f^*),\mbox{ with } \f \mbox{ an EmtbRE},
\end{array}$\smallskip

$
\begin{array}{rl}
\E'_{1,n} \mbox{ is a EmtbRE with }& \BB_1^n \mbox{ define the set of Bar operators, }\\
& \TT_1^n \mbox{ define the set of Tilde operators, } \\
 & \mbox{ and } \E_{1,n} \mbox{ a concatenation product of EmtbREs}.
 \end{array}$
 \end{df}
  The EmtbRE  $\E'_{i,j}$ is deduced from the expression
$\E'_{1,n}$ by taking as set of bars the subset $\BB_i^j=\{(k_1, k_2) \in \BB_1^n~|~i\leq k_1\leq k_2\leq j \}$ of $\BB_1^n$
and as set of tildes the subset  $\TT_i^j=\{(k_1,k_2) \in \TT_1^n~|~i\leq k_1\leq k_2\leq j \}$
of $\TT_1^n$.
 The {\it size} of $\E'_{1,n}$ denoted $|\E'_{1,n}|$ is the 
size of $\E_{1,n}$ added with the term $|\TT_1^n|+|\BB_1^n|$. 
 The {\it alphabetic width} $||\E'_{1,n}||$ of $\E'_{1,n}$ is the number of occurrences of letters in the expression.
\begin{example}
Consider the regular expression $\E_{1,5}=\E_1\cdot \E_2 \cdot \E_3\cdot \E_4 \cdot \E_5$.
 Let us consider the set of bars 
$\BB_1^5=\{(2,3),(3,5)\}$
and the set of tildes $\TT_1^5=\{(1,2),(4,5)\}$.
The EmtbREs $\E'_{1,5}$ and $\E'_{1,3}$ can be represented graphically as follows:\\
\small{\begin{tabular}{lr}
\begin{pspicture}(0,0)(2,0.8)
\psline[linewidth=0.9pt,linestyle=solid]{}(2.6,0.55)(4.4,0.55)
\psline[linewidth=0.9pt,linestyle=solid]{}(2,0.45)(3.3,0.45)

\pszigzag[coilwidth=0.07,coilheight=2,coilarm=0](1.5,0.3)(2.6,0.3)
\pszigzag[coilwidth=0.07,coilheight=2,coilarm=0](3.3,0.3)(4.4,0.3)

\uput{0.5cm}[1](0,0){$\E'_{1,5}=\E_1\cdot \E_2 \cdot \E_3\cdot \E_4 \cdot \E_5$}
\end{pspicture}
& 
\begin{pspicture}(-4,0)(2,0.8)
\psline[linewidth=0.9pt,linestyle=solid]{}(2,0.45)(3.1,0.45)

\pszigzag[coilwidth=0.07,coilheight=2,coilarm=0](1.5,0.3)(2.6,0.3)

\uput{0.5cm}[1](0,0){$\E'_{1,3}=\E_1\cdot \E_2 \cdot \E_3$}
\end{pspicture}
\end{tabular} 
}
 \end{example}                                 
 
\section{The language of a multi-tilde-bar expression}
The original semantical definition of the language of an EmtbRE \cite{CCMacta} is based on the description of how words are generated by overlapping tildes and bars.
Our approach is different: we provide a recursive definition of the language of an EmtbRE.
\begin{df}\label{df1}
Let $\E'_{1,n}$ be a  multi-tilde-bar expression. The language associated with
$\E'_{1,n}$ is recursively defined as follows:\smallskip

 \hspace*{3cm}{\footnotesize$\LL(\E'_{i,j})=\begin{cases}
   
                     {\cal L} \cup \{\varepsilon\} \mbox{~~~if~} (i,j)\in\TT_1^n, \\
                      {\cal L}\setminus \{\varepsilon\} \mbox{~~~if~} (i,j)\in\BB_1^n, \\
                   {\cal L}  \mbox{~~~~~~~~~~~~otherwise.} 
                   \end{cases}
$}\\
With {\footnotesize${\cal L}=\bigcup\limits_{k=i}^{j-1}\LL(\E'_{i,k})\cdot \LL(\E'_{k+1,j})$}~~~and ~~~{\footnotesize
$\LL(\E'_{k,k})=\begin{cases}
                      \LL(\E_{k,k})\cup \{\varepsilon\} \mbox{~~~if~} (k,k)\in\TT_1^n, \\
                      \LL(\E_{k,k})\setminus \{\varepsilon\} \mbox{~~~if~} (k,k)\in\BB_1^n, \\
                       \LL(\E_{k,k}) \mbox{~~~~~~~~~~~~otherwise.~}
                   \end{cases}
$}  \\for all $1\leq i<j\leq n$.
\end{df}
\begin{cor}\label{regular}
The language of a multi-tilde-bar expression $\E'_{1,n}$ is regular. 
\end{cor}
As we will see in the following, this recursive definition will allow us to provide the construction of the Glushkov automaton of
 any EmtbRE. It is worthwhile noticing that in \cite{CCMacta}, this construction is restricted to saturated EmtbREs,
that is expressions such that in each EmtbRE subexpression every  factor is equipped with either a tilde or a bar.\smallskip

Let us define a particular concatenation operator, denoted by $\odot_\varepsilon$,
as follows:\smallskip

{\footnotesize$
 \LL(\E'_{1,j})\odot_\varepsilon\LL(\E'_{j+k,n})=\left\{\begin{array}{ll}
                                                    \Big(\LL(\E'_{1,j})\cdot\LL(\E'_{j+k,n})\Big)\setminus\{\varepsilon\} & \mbox{ if }(1,n)\in \BB_1^n,\\
                                                     \Big(\LL(\E'_{1,j})\cdot\LL(\E'_{j+k,n})\Big) & \mbox{ otherwise}.
                                                   \end{array}\right.
$}
\begin{prop}\label{pr2}
Let $\E'_{1,n}$ be an EmtbRE. The language associated with $\E'_{1,n}$ can be recursively computed  as follows: \smallskip

{\footnotesize
$\LL(\E'_{1,k})= \Big(\LL(\E'_{1,k-1})\odot_\varepsilon\LL_k\Big)\cup\Big(\bigcup\limits_{j=1}^{k-1}\LL(\E'_{1,j})\odot_\varepsilon\nul(\E'_{j+1,k})\Big)\cup 
\nul(\E'_{1,k}),~{\scriptsize\forall 1<k\leq n},
$} \\
with {\footnotesize$\LL_i=\LL(\E'_{i,i})\cup\nul(\E'_{i,i})$, $\forall 1\leq i\leq n$}.
\end{prop}
\begin{proof}The proof is by induction on $k$, {\it i.e.} the number of factors in $\E'_{1,k}$. Let us consider the case where 
$k=2$. It is easy to prove that the proposition is true: \smallskip

{\footnotesize$\LL(\E'_{1,2})= \Big(\LL(\E'_{1,1})\odot_\varepsilon\LL_2\Big)\cup\Big(\LL(\E'_{1,1})\odot_\varepsilon\nul(\E'_{2,2})\Big)\cup
\nul(\E'_{1,2})
$}\\

We now suppose that the proposition
is satisfied for the  EmtbE $\E'_{1,k-1}$ and  we prove it is satisfied for
  $\E'_{1,k}$.\smallskip
  
{\footnotesize$\begin{array}{lcl}
\LL(\E'_{1,k})&\stackrel{Def.~\ref{df1}}{=}&\bigcup\limits_{j=1}^{k-1}\LL(\E'_{1,j})\odot_\varepsilon \LL(\E'_{j+1,k})\cup\nul(\E'_{1,k})
                    \\                          
            &=& \Big(\LL(\E'_{1,k-1})\odot_\varepsilon\LL_k\Big)\cup\Big(\bigcup\limits_{j=1}^{k-2}\LL(\E'_{1,j})\odot_\varepsilon\LL(\E'_{j+1,k})\Big)
              \cup\nul(\E'_{1,k})\\
          &\stackrel{Ind. Hyp.}{=}&\Big(\LL(\E'_{1,k-1})\odot_\varepsilon\LL_k\Big)\cup\nul(\E'_{1,k})\\
  &&\bigcup\limits_{j=1}^{k-2}\LL(\E'_{1,j})\odot_\varepsilon
              \Bigg(
             \Big(\LL(\E'_{j+1,k-1})\odot_\varepsilon\LL_{k}\Big)\\
              && \hspace*{2cm}\cup\Big(\bigcup\limits_{l=j+1}^{k-1}\LL(\E'_{j+1,l})\odot_\varepsilon
              \nul(\E'_{l+1,k})\Big)\cup \nul(\E'_{j+1,k})
              \Bigg)
         \end{array}$}\\
                       
    \hspace*{-0.5cm}{\footnotesize$\begin{array}{lcl} 
     &=& \hspace*{-0.4cm}\Big(\LL(\E'_{1,k-1})\odot_\varepsilon\LL_k\Big)\cup\nul(\E'_{1,k})
       \cup\Big(\bigcup\limits_{j=1}^{k-2}\LL(\E'_{1,j})\odot_\varepsilon \LL(\E'_{j+1,k-1})\odot_\varepsilon\LL_{k}\Big) \\
      && \cup\Big(\bigcup\limits_{j=1}^{k-2}\LL(\E'_{1,j})\odot_\varepsilon\big(\bigcup\limits_{l=j+1}^{k-1}\LL(\E'_{j+1,l})\odot_\varepsilon
      \nul(\E'_{l+1,k})\big)\Big) \\
       &&     \cup \Big(\bigcup\limits_{j=1}^{k-2}\LL(\E'_{1,j})\odot_\varepsilon\nul(\E'_{j+1,k})\Big)\\
       &\stackrel{Def.~\ref{df1}}{=}& \Big(\LL(\E'_{1,k-1})\odot_\varepsilon\LL_k\Big)\cup\nul(\E'_{1,k})\cup
       \Big(\LL(\E'_{1,k-1})\odot_\varepsilon\LL_{k}\Big) \\
      && \cup\Big(\bigcup\limits_{j=1}^{k-2}\bigcup\limits_{l=j+1}^{k-1}\big(\LL(\E'_{1,j})\odot_\varepsilon \LL(\E'_{j+1,l})\big)\odot_\varepsilon
      \nul(\E'_{l+1,k})\Big) \\
     &&\cup \Big(\bigcup\limits_{l=1}^{k-2}\LL(\E'_{1,l})\odot_\varepsilon\nul(\E'_{l+1,k})\Big)\\
     
 \end{array}
$}\smallskip

\noindent~A straightforward consequence of the Definition~\ref{df1} is that 
{\footnotesize$\LL(\E'_{1,j})\odot_\varepsilon\LL(\E'_{j+1,l})
\subseteq \LL(\E'_{1,l})$}, for all $1\leq j\leq k-2$. As a consequence we have:\smallskip

\hspace*{-0.5cm}${\footnotesize\begin{array}{lcl}
\LL(\E'_{1,k})&=&  \Big(\LL(\E'_{1,k-1})\odot_\varepsilon\LL_k\Big)\cup\Big(\LL(\E'_{1,k-1})\odot_\varepsilon
      \nul(\E'_{k,k})\Big)\cup \nul(\E'_{1,k})  \\
&&\hspace*{-0.45cm}\cup 
\Big(\bigcup\limits_{l=1}^{k-2}\LL(\E'_{1,l})\odot_\varepsilon\nul(\E'_{l+1,k})\Big) \cup\Big(\big(\LL(\E'_{1,k-2})\odot_\varepsilon
\LL(\E'_{k-1,k-1})\big)
\odot_\varepsilon
      \nul(\E'_{k,k})\Big) 
\end{array}}$ \smallskip

Finally, {\footnotesize $\LL(\E'_{1,k})=  \Big(\LL(\E'_{1,k-1})\odot_\varepsilon\LL_k\Big)\cup
\Big(\bigcup\limits_{l=1}^{k-1}\LL(\E'_{1,l})\odot_\varepsilon\nul(\E'_{l+1,k})\Big) 
\cup \nul(\E'_{1,k})$}

 \cqfd
\end{proof}
\section{The position automata of a multi-tilde-bar  expression}
 
\subsection{Glushkov functions for a regular expression}
Let $\E$ be a regular expression.
In order to construct a non-deterministic finite automaton
recognizing $\LL(\E)$, Glushkov \cite{Glush} and
McNaughton-Yamada \cite{MnYa} have introduced independently the so-called {\it position automaton}.
Given a regular expression $\E$ in linearized form,
the following sets called {\it Glushkov functions} are defined as follows,
where $x\in\pos(\E)$ and $u,v\in \pos(\E)^*$:
\begin{eqnarray*}
\First(\E)&=& \{x\in \pos(\E)~|~ xv\in\LL(\E)\}\\
\Last(\E)&=& \{x\in \pos(\E)~|~ ux\in\LL(\E)\}\\
\Follow(x,\E)&=& \{y\in \pos(\E)~|~ uxyv\in\LL(\E)\}
\end{eqnarray*}
The position automaton of $\E$ is deduced from these position sets as follows.
We first add a specific position $q_0$ to the set $\pos(\E)$ and
we set $\pos_0(\E) = \pos(\E)\cup\{q_0\}$; the set $\Last_0 (\E)$ is equal to $\Last(\E)$ if $\nul(\E) =\emptyset$
and to $\Last(\E) \cup\{q_0\}$ otherwise; the set $\Follow_0 (x,\E)$ is equal to $\Follow(x,\E)$ if
$x \in \pos(\E)$ and to $\First(\E)$ if $x = q_0$.

The position automaton ${\cal P}_{\E}$ of a regular expression $\E$ is defined by the 5-uple \\$ \langle \pos_0(\E), A, \delta, q_0, \Last_0(\E)\rangle$  such that: \\
$~~~~~~~~~~~~~~\delta(x, a) = \{y~|~y \in \Follow_0 ( x,\E)\mbox{ and }h(y) = a\},~\forall x \in \pos_0 (\E),~\forall a\in A.$

The position automaton ${\cal P}_{\E}$ recognizes the language $\LL(\E)$ \cite{Glush,MnYa}.

Glushkov functions can be defined for bar expressions and tilde expressions as follows,
where $x\in \pos(\E)$:
$$\begin{array}{lllll}
 \First(\E)&=&\First(\overline{\E})&=&\First(\begin{pspicture}(0,0)(0.4,0.35)
\pszigzag[coilwidth=0.07,coilheight=2,coilarm=0](0,0.32)(0.4,0.32)
\uput{0.1cm}[1](0,0.1){$\E$}
\end{pspicture}),\\ 
 \Last(\E)&=&\Last(\overline{\E})&=&\Last(\begin{pspicture}(0,0)(0.4,0.35)
\pszigzag[coilwidth=0.07,coilheight=2,coilarm=0](0,0.32)(0.4,0.32)
\uput{0.1cm}[1](0,0.1){$\E$}
\end{pspicture}),\\
 \Follow(x,\E)&=&\Follow(x,\overline{\E})&=&\Follow(x,\begin{pspicture}(0,0)(0.4,0.35)
\pszigzag[coilwidth=0.07,coilheight=2,coilarm=0](0,0.32)(0.4,0.32)
\uput{0.1cm}[1](0,0.1){$\E$}
\end{pspicture}).
 \end{array}$$
As a consequence the computation of Glushkov functions can be extended to the family of EmtbREs.
Such an extension is described in \cite{CCMacta}; it addresses the subfamily of saturated EmtbREs for which
every factor is equipped with either a tilde or a bar.

\subsection{Glushkov functions for a multi-tilde-bar expression}
In this section, we  address the general case:  
we show how to compute the Glushkov functions of an EmtbRE
for which there is no restriction on the distribution of tilde and bar operators over the factors of the expression.
\begin{prop}\label{flf}
Let $\E=\E_1\cdot \E_2 \cdots \E_n$, with $n\geq 1$,
and $\E'=\E'_{1,n}$ be an EmtbRE in linearized form.
 
Let $k$ be an integer such that $1\leq k\leq n$ and $x$ be a position in $\pos(\E_k)$. 
The Glushkov functions associated with $\E'$ are recursively computed according to the following formulas:
\vspace*{-0.5cm}
\begin{eqnarray*}
  \pos(\E'_{1,n}) &=&\bigcup_{k=1}^{n}\pos(\E_k)
  \end{eqnarray*} 
\scalebox{0.8}{
$  \begin{array}{|c||c|c|c|}
\hline\hline
{\bf \E'}&{\bf \First(\E')} & {\bf \Last(\E')} & {\bf \Follow(x,\E')} \\
\hline\hline 
\emptyset,\varepsilon & \emptyset &\emptyset &\emptyset \\
\hline
a\in A & a &a&\emptyset \\
\hline
\f+\g & \First(\f)\cup\First(\g)  &  \Last(\f)\cup\Last(\g) & \Follow(x,\f)\cup\Follow(x,\g)  \\
\hline
 & & & \Follow(x,\f)\cup\First(\g), \mbox{\bf  if }x\in\Last(\f),\\
\f\cdot \g  & \First(\f)\cup \nul(\f)\First(\g)  & \Last(\g)\cup\nul(\g)\Last(\f)   &  \Follow(x,\f),\mbox{ \bf if }x\in\pos(\f)\setminus \Last(\f),\\
  &   &    & \Follow(x,\g),\mbox{\bf if }x\in\pos(\g). \\
\hline
\f^* &\First(\f) &\Last(\f) & \Follow(x,\f),\mbox{ \bf if }x\in\Last(\f),\\
 &  &  & \Follow(x,\f)\cup \First(\f),\mbox{ \bf otherwise. }\\
 \hline
\end{array}$}
\begin{eqnarray}
 \First(\E'_{1,n}) &=& \First(\E_{1})\uplus \biguplus_{j=1}^{n-1}\nul(\E'_{1,j})\cdot\First(\E'_{j+1,j+1}),\label{first}  \\
  \Last(\E'_{1,n}) &=&\Last(\E_{n})\uplus \biguplus_{j=1}^{n-1}\nul(\E'_{j+1,n})\cdot\Last(\E'_{j,j}), \label{last} \\ 
   \Follow(x,\E'_{1,n})&=& \left\{\begin{array}{ll} \Follow(x,\E_k) & \mbox{if}~{\footnotesize(k=n)\vee (x\notin \Last(\E_k)),\label{follow}} \\
                            \Follow(x,\E_k)\uplus \First(\E'_{k+1,n}) & \mbox{otherwise.}
                              \end{array}\right.  
 \end{eqnarray}
\end{prop}
\begin{proof}~Proof is restricted to the non-classical cases:\\
\eqref{first} from the definition of the function $\First$, one has:\\
$\First(\E'_{1,n})= \{x\in \pos(\E'_{1,n})~|~ xv\in\LL(\E'_{1,n})\}$. Using the Proposition~\ref{pr2} and by induction on $n$, one can deduce the following equalities:\\
$\begin{array}{rcl}
 \First(\E'_{1,n})&=&\scriptsize{\bigg\{x\in \pos(\E'_{1,n})~|~ xv\in\Big(\LL(\E'_{1,n-1})\cdot\LL_n\Big)\bigg\}}\\
  && \cup \scriptsize{\bigg\{x\in \pos(\E'_{1,n})~|~ xv\in\Big(\bigcup\limits_{j=1}^{n-1}\LL(\E'_{1,j})\cdot\nul(\E'_{j+1,n})\Big)\bigg\}}
\\ 
   &=&\First(\E'_{1,n-1})\cup\nul(\E'_{1,n-1})\cdot\First(\E'_{n,n})\cup\bigcup\limits_{j=1}^{n-1}\First(\E'_{1,j})\\
    &=&\nul(\E'_{1,n-1})\cdot\First(\E'_{n,n})\cup\bigcup\limits_{j=2}^{n-1}\First(\E'_{1,j}) \cup\First(\E'_{1,1})\\
    &\stackrel{Ind. Hyp.}{=}& \First(\E'_{1,1})\uplus\biguplus\limits_{j=1}^{n-1}\nul(\E'_{1,j})\cdot\First(\E'_{j+1,j+1})
\end{array}
$
\\
\eqref{last} from the definition of the function $\Last$, one has:\\
$\Last(\E'_{1,n})= \{x\in \pos(\E'_{1,n})~|~ xv\in\LL(\E'_{1,n})\}$. Using the Proposition~\ref{pr2} and by induction on $n$
one can deduce the following equalities:\\
$\begin{array}{rcl}
 \Last(\E'_{1,n})&=&\scriptsize{\bigg\{x\in \pos(\E'_{1,n})~|~ vx\in\Big(\LL(\E'_{1,n-1})\cdot\LL_n\Big)\bigg\}}\\
  &&\cup\scriptsize{\bigg\{x\in \pos(\E'_{1,n})~|~ vx\in\Big(\bigcup\limits_{j=1}^{n-1}\LL(\E'_{1,j})\cdot\nul(\E'_{j+1,n})\Big)\bigg\}}\\
    &=&\Last(\E'_{n,n})\cup\Last(\E'_{1,n-1})\cdot\nul(\E'_{n,n})\cup\bigcup\limits_{j=1}^{n-1}\Last(\E'_{1,j})\cdot\nul(\E'_{j+1,n})\\
  &=&\Last(\E'_{n,n})\cup\bigcup\limits_{j=1}^{n-1}\Last(\E'_{1,j})\cdot\nul(\E'_{j+1,n})
  \end{array}$\\
  $\begin{array}{rcl}
 \Last(\E'_{1,n})  &\stackrel{Ind. Hyp.}{=}&\Last(\E'_{n,n})\uplus \biguplus\limits_{j=1}^{n-1}\Last(\E'_{j,j})\cdot\nul(\E'_{j+1,n})
\end{array}
$
\\
(\ref{follow}) proof is similar as for \eqref{first} and \eqref{last}.
 \cqfd
\end{proof}
\begin{cor}\label{mtbNULL} 
The Glushkov functions of a multi-tilde-bar expression can be written as a disjoint union which involves the $\First$, $\Last$, and $\Follow$ sets associated
 with sub-expressions of $\E'_{1,n}$ (not of $\E_{1,n}$) and the value of the function $\nul(\E'_{i,j})$ for all $1\leq i\leq j\leq n$.
\end{cor}
 The following proposition can be deduced from the Definition~\ref{df2}.
\begin{prop}\label{propnul}Let $\E'$ be an EmtbRE in linearized form. The function $\nul(\E')$ can be recursively computed as follows: \\

$~~~~~~~\begin{array}{lll}
\nul(\emptyset)& = & \emptyset ,\\
\nul(\varepsilon) &=& \{\varepsilon\}, \\
\nul(a) &=& \emptyset,
\end{array}~~~~~~~~~~~~~~
\begin{array}{rll}
\nul(\f+\g) & = & \nul(\f)\cup\nul(\g) , \\
\nul(\f\cdot \g) & = & \nul(\f)\cdot\nul(\g), \\
 \nul(\f^*) & = & \{\varepsilon\} ,
 \end{array}$
\begin{eqnarray}\label{nul}
\nul(\E'_{1,n})&=& \begin{cases}
\emptyset & \mbox{~if ~}(1,n) \in \BB_1^n,\\
\{\varepsilon\} & \mbox{~if ~}(1,n) \in \TT_1^n,\\
\bigcup\limits_{j=1}^{n-1}\nul(\E'_{1,j})\cdot\nul(\E'_{j+1,n}) & \mbox{~otherwise.}
\end{cases}
\end{eqnarray}
\end{prop}
\begin{proof}
Proof is by induction on the size of $\E$. It is restricted to the non-classical case (\ref{nul}).\\
If $(1,n)\in\TT_1^n$, then $\E'_{1,n}$ can be written as $\overline{\f}$. Thus, by the definition of the set $\nul$, we have 
$\nul(\E'_{1,n})=\{\varepsilon\}$.  If $(1,n)\in\BB_1^n$, then $\E'_{1,n}$ can be written as \begin{pspicture}(0,0)(0.4,0.35)
\pszigzag[coilwidth=0.07,coilheight=2.2,coilarm=0](0,0.32)(0.4,0.32)
\uput{0.1cm}[1](0,0.1){$\f$}
\end{pspicture}.  Thus,  by the definition of the set $\nul$, we have  $\nul(\E'_{1,n})=\emptyset$.\\
Let us suppose  that $(1,n)\notin\TT_1^n\cup\BB_1^n $,
one has:\\
$\begin{array}{lcl}
 \varepsilon\in\LL(\E'_{1,n})&\stackrel{Def.~\ref{df1}}{\Leftrightarrow}&\varepsilon\in\bigcup\limits_{j=1}^{n-1}\LL(\E'_{1,j})\cdot \LL(\E'_{j+1,n}) \\

  &\Leftrightarrow & \varepsilon\in\bigcup\limits_{j=1}^{n-1}\big(\nul(\E'_{1,j})\cdot\nul(\E'_{j+1,n})\big)\\  
 &\Leftrightarrow & \varepsilon\in\nul(\E'_{1,n})
\end{array}
$\\

\cqfd
\end{proof}
\begin{example}\label{example1}
Let us consider the following EmtbRE: \smallskip

\hspace*{1cm}{\small\begin{pspicture}(0,0)(2,0.9)
\psline[linewidth=0.9pt,linestyle=solid]{}(2.6,0.3)(5.5,0.3)
\psline[linewidth=0.9pt,linestyle=solid]{}(4,0.43)(6.9,0.43)

\pszigzag[coilwidth=0.065,coilheight=2.5,coilarm=0](6.9,0.6)(8,0.6)
\pszigzag[coilwidth=0.065,coilheight=2.5,coilarm=0](1.5,0.6)(4,0.6)
\pszigzag[coilwidth=0.065,coilheight=2.5,coilarm=0](2.5,0.73)(7.6,0.73)
\uput{0.5cm}[1](0,0){$\E'_{1,7}=a_1^*\cdot b_2\cdot (c_3+\varepsilon) \cdot (d_4+\varepsilon)\cdot (e_5+\varepsilon)\cdot f_6\cdot g_7^*$}
\end{pspicture}}

\bigskip
                                                           
 The language associated with $\E'_{1,7}$ is: \smallskip
 
 ${\footnotesize\{ a_1b_2, a_1b_2g_7, b_2,b_2c_3, b_2g_7,\cdots, b_2e_5,b_2c_3d_4e_5,b_2c_3d_4e_5f_6,b_2c_3f_6g_7,}$\\
${\footnotesize~~~~~~~~~b_2d_4e_5f_6g_7,\cdots,d_4,d_4e_5,d_4e_5f_6,d_4e_5f_6g_7,\cdots,e_5,e_5f_6,e_5f_6g_7,\cdots\}}$\smallskip

The associated Glushkov functions are: \smallskip

\begin{minipage}{6cm}
\footnotesize{ $\begin{array}{rcl}
  \pos(\E')&=&\{a_1,b_2,c_3,d_4,e_5,f_6,g_7\} \\
 \nul(\E')&=& \emptyset\\
 \First(\E')&=&\{a_1,b_2,d_4,e_5\}\\
 \Last(\E')&=&\{b_2,c_3,d_4,e_5,f_6,g_7\}\\
 \Follow(a_1,\E')&=&\{a_1,b_2\}\\
  \Follow(b_2,\E')&=& \{c_3,d_4,g_7\}\\
\Follow(c_3,\E')&=& \{d_4,e_5,f_6\}\\
 \Follow(d_4,\E')&=& \{e_5,f_6\}\\
  \Follow(e_5,\E')&=& \{f_6\}\\
 \Follow(f_6,\E')&=& \{g_7\}\\
  \Follow(g_7,\E')&=& \{g_7\}
 \end{array}
$ }\end{minipage}
\begin{minipage}{6cm}
\begin{figure}[H]
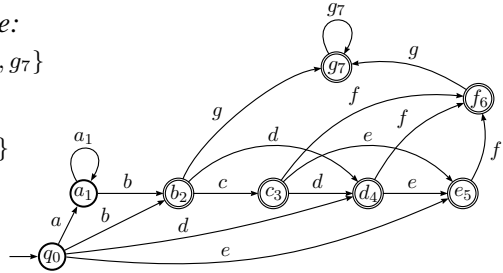

\scalebox{0.7}{\VCDraw{%
\begin{VCPicture}{(-1,-2)(20,4)}

\State[q_0]{(0,-1)}{Q} \State[a_1]{(1,1)}{A} \FinalState[b_2]{(4,1)}{B}\FinalState[c_3]{(7,1)}{C}\FinalState[d_4]{(10,1)}{D}
\FinalState[e_5]{(13,1)}{E}\FinalState[f_6]{(13.5,4)}{F} \FinalState[g_7]{(9,5)}{G}

\Initial{Q}
\EdgeL{Q}{A}{a} \EdgeL{A}{B}{b} \EdgeL{B}{C}{c} \EdgeL{C}{D}{d} \EdgeL{D}{E}{e} 

\EdgeL{Q}{B}{b} 
\VArcL{arcangle=-5}{Q}{D}{d} 
\VArcL{arcangle=-15}{Q}{E}{e}

\VArcL{arcangle=30,ncurv=.5}{B}{G}{g}
\VArcL[.5]{arcangle=45,ncurv=.8}{B}{D}{d}
\VArcL[.5]{arcangle=45,ncurv=.8}{C}{E}{e}
\VArcL[.5]{arcangle=20,ncurv=.8}{D}{F}{f}
\VArcL[.5]{arcangle=30,ncurv=.8}{C}{F}{f}
\VArcR[.5]{arcangle=-20,ncurv=.8}{F}{G}{g}
\VArcR[.5]{arcangle=-20,ncurv=.8}{E}{F}{f}
\LoopN[.5]{A}{a_1}
\LoopN[.5]{G}{g_7}

\end{VCPicture}
}}
\caption{The Position automaton ${\cal A}_{\E'_{1,7}}$}
\end{figure}
\end{minipage}             
                               
\end{example}
\section{Efficient computations of the position automaton and of the c-continuation automaton}
In this section, we present efficient algorithms to compute the Glushkov functions of a multi-tilde-bar expression $\E'$, based on the formulas of the Proposition~\ref{flf}.
According to the Corollary~\ref{mtbNULL}, the worst case time complexity of these algorithms depends on the worst case time complexity of the function $\nul(\E')$
that we first study.
 \subsection{Computation of $\nul(\E')$}
According to the Proposition~\ref{propnul}, a naive computation of the function $\nul$ of the EmtbRE $\E'_{1,n}$ can be performed
 using the following Algorithm.
 
 \scalebox{0.8}{
\RestyleAlgo{boxed}
  \begin{algorithm}[H]
  \DontPrintSemicolon
\SetAlgoVlined
\SetAlgoLined
 \KwData{$\E'_{i,j}$}
 \KwResult{$\nul(\E'_{i,j})$ }
 \BlankLine
 \For{$i\leftarrow 1$ \KwTo $n$}{
     \eIf{$(i,i)\in \BB_1^n$}{
    $\nul(\E'_{i,i})=\emptyset$\;
    }{  \eIf{$(i,i)\in \TT_1^n$}{$\nul(\E'_{i,i})=\{\varepsilon\}$\;
    }
    {$\nul(\E'_{i,i})=\nul(\E_{i,i})$\;
         }
     
     }
 }
\For{$k\leftarrow 1$ \KwTo $n-1$}{
     \For{$i\leftarrow 1$ \KwTo $n-k$}{
          \eIf{$(i,i+k)\in \BB_1^n$}{
    $\nul(\E'_{i,i+k})=\emptyset$\;
    }{  \eIf{$(i,i+k)\in \TT_1^n$}{$\nul(\E'_{i,i+k})=\{\varepsilon\}$\;
    }
    {$\nul(\E'_{i,i+k})=\bigcup\limits_{j=1}^{i+k-1}\nul(\E'_{i,j})\cdot\nul(\E'_{j+1,i+k})$\;
         }
     
     }
 }
 } 
 \end{algorithm}
}\smallskip

The different steps of the algorithm are illustrated through the following example.
 \begin{example}
 Consider the EmtbRE $\E'_{1,3}$ such that {\footnotesize$\TT_1^3=\{(1,1),(2,3)\}$, $\BB_1^3=\{(1,2),(3,3)\}$}, and $ \E_1=a$, $\E_2=(b+\varepsilon)$,
 $\E_3=(c+\varepsilon)$. The diagram below is a graphical representation of  the recursive dependency between different values of $\nul(\E'_{i,j})$.

 \scalebox{0.85}{

\hspace*{2cm}\scalebox{0.7}{\VCDraw{%
\begin{VCPicture}{(3.5,-1.5)(7,7)}
\LargeState

 \StateVar[\E'_{1,4}]{(6,6)}{A}
\StateVar[\E'_{1,3}]{(4,4)}{B}\StateVar[\E'_{2,4}]{(8,4)}{K}
\StateVar[\E'_{3,4}]{(10,2)}{E}
\StateVar[\E'_{2,2}]{(4,0)}{G}
\StateVar[\E'_{4,4}]{(12,0)}{I}

 \SetStateLineColor{blue}
\SetStateLabelColor{blue}
\StateVar[\E'_{1,2}]{(2,2)}{C}
\StateVar[\E'_{3,3}]{(8,0)}{H}

\SetStateLineColor{green}
\SetStateLabelColor{green}
\StateVar[\E'_{1,1}]{(0,0)}{F}
\StateVar[\E'_{2,3}]{(6,2)}{D}


\EdgeR{B}{A}{} 
\EdgeL{K}{A}{} 
\EdgeL{C}{B}{} 
\EdgeL{D}{B}{} 
\EdgeL{D}{K}{} 
\EdgeL{E}{K}{} 
\EdgeL{F}{C}{} 
\EdgeL{G}{C}{}
\EdgeL{G}{D}{} 
\EdgeL{H}{D}{}
\EdgeL{H}{E}{} 
\EdgeL{I}{E}{}
\end{VCPicture}
}
}
}
\hspace*{3cm}\begin{tabular}{|l|c|l|} 
\hline
EmtbRE & $\nul $ & \\
\hline
$\E'_{1,1}$ & $\{\varepsilon\}$ & $(1,1) \in \TT_1^3$ \\
$\E'_{2,3}$ & $\{\varepsilon\}$ & $(2,3) \in \TT_1^3$ \\
$\E'_{1,2}$ & $\emptyset$ & $(1,2) \in \BB_1^3$ \\
$\E'_{3,3}$ & $\emptyset$ & $(3,3) \in \BB_1^3$ \\
\hline
\end{tabular}

It holds: \\ $\begin{array}{llcccc}
\nul(\E'_{1,3}) &= &\big(\nul(\E'_{1,1})\cdot \nul(\E'_{2,3})\big)& \cup &\big(\nul(\E'_{1,2})\cdot
\nul(\E'_{3,3})\big)\\
 &=& \{\varepsilon\}   & \cup &\emptyset \\
&=  &\{\varepsilon\}&
 \end{array}$
 \end{example}
Let us consider the case of an EmtbRE $\E'_{1,n}$.
There are $(n-k)$ vertices on the $k^{th}$ line, corresponding to tilde or bar operators $(1,1+k),(2,2+k),\dots$
 The computation of the associated $\nul(\E'_{i,i+k})$ functions requires:
 \begin{itemize}
  \item a constant number of elementary test operations: \\  if {$(i,i+k)\in \TT_1^n$ or  $(i,i+k)\in \BB_1^n$},
 \item $(k-1)$ concatenations of {$\big(\nul(\E'_{i,j})\cdot\nul(\E'_{j+1,i+k})\big)$},
\item $(k-2)$ unions.
 \end{itemize}
Finally,  $\sum\limits_{k=2}^{n}2*k*(n-k+1)$ operations are needed to
   compute the function $\nul(\E'_{1,n})$. 
\begin{prop}
 Let $\E'_{1,n}$ be an EmtbRE. The function $\nul(\E'_{1,n})$ can be computed in $O(|\E'_{1,n}|+n^3)$ time.
\end{prop}
Notice that the function $\nul(\E'_{1,n})$ can be computed
by making use of one of the numerous algorithms which compute the transitive closure of a DAG (see for example \cite{Chen}).
Although these algorithms have the same $O(n^3)$ worst case time complexity as the naive algorithm
they likely have a better running time performance than the naive algorithm.
  \subsection{Computation of the Glushkov functions}
 According to Corollary~\ref{mtbNULL}, for an EmtbRE $\E'_{1,n}$, the functions  {$\First(\E'_{1,n})$, (Resp. $\Last(\E'_{1,n})$)},  
 and  {\footnotesize$\Follow(x,\E'_{1,n})$}
  can be written as disjoint unions of some  {$\First(\E'_{i,i})$ (Resp. $\Last(\E'_{i,i})$)} sets. Thus,
  the following proposition holds.
  \begin{prop}
   Lets $\E'_{1,n}$ be an EmtbRE and $x\in\pos(\E'_{1,n})$. The functions $\First(\E'_{1,n})$, $\Last(\E'_{1,n})$, and $\Follow(x,\E'_{1,n})$ 
can be computed in $O(|\E'_{1,n}|+n^3)$ time.
  \end{prop}

\subsection{Computation of a c-continuation over a {\ZPC}-structure}
According to Corollary \ref{mtbNULL}, a multi-tilde-bar expression can be viewed as a standard regular expression equipped with
a specific computation for the function $\nul$.
The computation of the Glushkov functions of a multi-tilde-bar expression obviously depends on the definition of the function $\nul$:
for example, an alternative interpretation of the tilde operator can be associated with the following definition of $\nul$: \begin{center}
      $\nul(\begin{pspicture}(0,0)(2,0.35)
\pszigzag[coilwidth=0.07,coilheight=2.2,coilarm=0](0,0.4)(2,0.4)
\uput{0.1cm}[1](0,0.05){$\E_1\cdot \E_2\cdots \E_n$}
\end{pspicture}~~~) = \{\varepsilon\} \Leftrightarrow \big(\varepsilon \in \LL(\E_1)\big)\wedge \big(\varepsilon \in \LL(\E_n)\big)$.
     \end{center}
The \ZPC-structure \cite{ZPC} can be extended to multi-tilde-bar expressions in a natural way (see Figure \ref{zpc}), 
by representing the tilde and bar operators by edges connecting the $'\cdot'$-nodes of the product.
Therefore, all the algorithms based on the \ZPC-structure, {\it i.e.} the construction of the c-continuation automaton \cite{CZ}, of the equation automaton \cite{CZ},
of the follow automaton~\cite{CNZ} and of the weighted position automaton \cite{CLOZ} also work for multi-tilde-bar expressions. \\
Moreover the worst case time complexity in the case of multi-tilde-bar expressions is the worst case time complexity of the standard case augmented with the worst case time complexity of the function $\nul$.
Therefore, the following theorem can be stated.
\begin{thm}
Let $\E'$ be a multi-tilde-bar expression and ${\cal N}$ the worst case time complexity of the function $\nul$.
The position automaton, the c-continuation automaton, the follow automaton and the equation automaton 
associated with $\E'$ can be computed in \\ $O(|\E'|\times||\E'||+{\cal N})$ time.
\end{thm}
The computation of a c-continuation through a \ZPC-structure is illustrated by the following example.

\begin{example}~Let us consider the following EmtbRE: \smallskip

\hspace*{1cm}{\small\begin{pspicture}(0,0)(2,0.6)
\psline[linewidth=0.9pt,linestyle=solid]{}(2.4,0.3)(5.2,0.3)
\psline[linewidth=0.9pt,linestyle=solid]{}(3.8,0.4)(6.6,0.4)

\pszigzag[coilwidth=0.064,coilheight=2.5,coilarm=0](2.4,0.55)(7.2,0.55)
\uput{0.5cm}[1](0,0){$\E'_{1,6}=\Bigg( a_1\cdot (b_2+\varepsilon) \cdot (c_3+\varepsilon)\cdot (d_4+\varepsilon)\cdot e_5\cdot
       f_6^*\Bigg)^*$}
\end{pspicture}}. \bigskip

Let us explain how to compute the c-continuation of $\E'$ associated with some position $x$, denoted by $c_{x}(\E')$.
The \ZPC-structure of $\E'$ is partially shown in Figure \ref{zpc}, with all the links which are necessary to computes $c_{a_1}(\E')$ and $c_{b_2}(\E')$ .\\
 On the right-hand side, the standard $\First$ tree is added with blue (resp. green) links between some $'\cdot'$-nodes  
 which represent bar (resp. tilde) operators over factors of $\E$.
 The edge connecting any $'\cdot'$-node to its right son is marked by the value of the function $\nul$ associated with
 its left son, and all other edges are marked by $\varepsilon$.\\
 On the left-hand side, the standard $\Last$ tree is added with blue (resp. green) links between $'\cdot'$-nodes  
 which represent bar (resp. tilde) operators over factors of $\E$.
 The edge connecting any $'\cdot'$-node to its left son is marked by the value of the function $\nul$ associated with
 its right son, and all other edges are marked by $\varepsilon$.\\
 The two trees are connected by the so-called $\Follow$ links (red links).
 For each $'\cdot'$-node,
 there is a Follow link going from its left son in the $\Last$ tree to its right son in the $\First$ tree,
 and for each $\star$-node,
 there is a Follow link going from its son in the $\Last$ tree to the $\star$-node itself in the $\First$ tree.\vspace*{1cm}

\begin{figure}
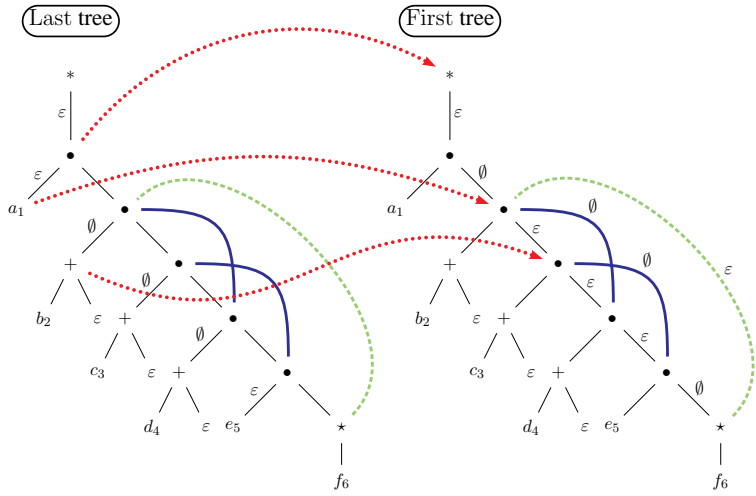


\scalebox{0.6}{\VCDraw{
\begin{VCPicture}{(4,21)(30,11)}
\LargeState\StateVar[\mbox{\Large$\Last$ tree}]{(10,23)}{x}
\LargeState\StateVar[\mbox{\Large$\First$ tree}]{(24,23)}{x}
\SetStateLineStyle{none}
\State[*]{(10,21)}{a}

\State[\bullet]{(10,18)}{c}

\State[a_1]{(8,16)}{e}
\State[\bullet]{(12,16)}{f}

\State[+]{(10,14)}{g}
\State[\bullet]{(14,14)}{h}

\State[b_2]{(9,12)}{i}              
\State[\varepsilon]{(11,12)}{j}
\State[+]{(12,12)}{k}
\State[\bullet]{(16,12)}{l}

\State[c_3]{(11,10)}{m}
\State[\varepsilon]{(13,10)}{w}
\State[+]{(14,10)}{n}
\State[\bullet]{(18,10)}{o}

\State[d_4]{(13,8)}{q}
\State[\varepsilon]{(15,8)}{p}
\State[e_5]{(16,8)}{r}
\State[\star]{(20,8)}{s}

\State[f_6]{(20,6)}{t}


\SetStateLineStyle{none}
\State[*]{(24,21)}{A}

\State[\bullet]{(24,18)}{C1}

\State[a_1]{(22,16)}{E}
\State[\bullet]{(26,16)}{F}

\State[+]{(24,14)}{G}
\State[\bullet]{(28,14)}{H}

\State[b_2]{(23,12)}{I}              
\State[\varepsilon]{(25,12)}{J}
\State[+]{(26,12)}{K}
\State[\bullet]{(30,12)}{L}

\State[c_3]{(25,10)}{M}
\State[\varepsilon]{(27,10)}{W}
\State[+]{(28,10)}{N}
\State[\bullet]{(32,10)}{O}

\State[d_4]{(27,8)}{Q}
\State[\varepsilon]{(29,8)}{P}
\State[e_5]{(30,8)}{R}
\State[\star]{(34,8)}{S}

\State[f_6]{(34,6)}{T}
\SetEdgeArrowStyle{-}

\EdgeR{a}{c}{\varepsilon}

\EdgeR{c}{e}{\Large\varepsilon}
\EdgeL{c}{f}{}
\EdgeR{f}{g}{\emptyset}
\EdgeL{f}{h}{}

\EdgeR{g}{i}{}
\EdgeL{g}{j}{}
\EdgeR{h}{k}{\emptyset}
\EdgeL{h}{l}{}

\EdgeL{k}{m}{}
\EdgeL{k}{w}{}
\EdgeR{l}{n}{\emptyset}
\EdgeL{l}{o}{}

\EdgeL{n}{q}{}
\EdgeL{n}{p}{}

\EdgeR{o}{r}{\varepsilon}
\EdgeL{o}{s}{}

\EdgeL{s}{t}{}



\SetEdgeArrowStyle{-}

\EdgeL{A}{C1}{\Large\varepsilon}

\EdgeL{C1}{E}{}
\EdgeL{C1}{F}{\emptyset}
\EdgeL{F}{G}{}
\EdgeL{F}{H}{\Large\varepsilon}

\EdgeL{G}{I}{}
\EdgeL{G}{J}{}
\EdgeL{H}{K}{}
\EdgeL{H}{L}{\Large\varepsilon}

\EdgeL{K}{M}{}
\EdgeL{K}{W}{}
\EdgeL{L}{N}{}
\EdgeL{L}{O}{\Large\varepsilon}

\EdgeL{N}{Q}{}
\EdgeL{N}{P}{}

\EdgeL{O}{R}{}
\EdgeL{O}{S}{\emptyset}

\EdgeL{S}{T}{}

\SetEdgeLineColor{YellowGreen}
\SetEdgeLineWidth{1pt}


\SetEdgeLineStyle{dashed}
\SetEdgeLineWidth{3pt}
\VArcL{arcangle=-90,ncurv=.8}{s}{f}{}
\VArcR{arcangle=-90,ncurv=.8}{S}{F}{\varepsilon}


\SetEdgeLineColor{Blue}
\SetEdgeLineWidth{3pt}
\SetEdgeLineStyle{solid}
\VArcL{arcangle=45,ncurv=1.2}{h}{o}{}
\VArcL{arcangle=45,ncurv=1.2}{f}{l}{}
\VArcL{arcangle=45,ncurv=1.2}{H}{O}{\emptyset}
\VArcL{arcangle=45,ncurv=1.2}{F}{L}{\emptyset}
\SetEdgeLineColor{Red}
\SetEdgeArrowStyle{->}
\SetEdgeArrowWidth{10pt}
\SetEdgeLineStyle{dotted}
\SetEdgeLineWidth{4pt}
\VArcL{arcangle=40}{c}{A}{}
\VArcL{arcangle=20}{e}{F}{}
\VCurveR{angleA=-25,angleB=160}{g}{H}{}
\end{VCPicture}
}
}
\caption{The \ZPC-structure associated with the multi-tilde-bar expression $\E'_{1,6}$.}\label{zpc}
\end{figure}
\vspace*{-0.5cm}
The computation of a c-continuation using a \ZPC-structure can be done in a similar way as in the standard case.
Let $<(l_1,r_1), (l_2,r_2), ...,(l_k,r_k)>$ the list of follow links in the path going from a position $x$ to the root of the Last tree.
Let us denote by $\f_{i}$ the subexpression associated with the node $r_i$ in the First tree. Then the c-continuation $c_x$ associated with $x$ 
is the expression $\f_{1}\cdots \f_{k}$. In our example we have:\\

{\footnotesize\hspace*{2cm} \begin{tabular}{l}
\begin{pspicture}(0,-0.2)(2,0.6)
\psline[linewidth=0.9pt,linestyle=solid]{}(2.1,0.3)(4.5,0.3)
\psline[linewidth=0.9pt,linestyle=solid]{}(3.2,0.4)(5.8,0.4)

\pszigzag[coilwidth=0.064,coilheight=2.5,coilarm=0](2,0.55)(6.4,0.55)
\uput{0.5cm}[1](0,0){$c_{a_1}(\E)=\Big((b_2+\varepsilon) \cdot (c_3+\varepsilon)\cdot (d_4+\varepsilon)\cdot e_5\cdot
       f_6^*\Big)\cdot \E'_{1,6}$}
\end{pspicture} 
\\
\begin{pspicture}(0,-0.1)(2,0.4)

\psline[linewidth=0.9pt,linestyle=solid]{}(2.1,0.3)(4.5,0.3)

\uput{0.5cm}[1](0,0){$c_{b_2}(\E)=\Big((c_3+\varepsilon)\cdot (d_4+\varepsilon)\cdot e_5\cdot f_6^*\Big)\cdot \E'_{1,6}$}

\end{pspicture}
\end{tabular}
}
\end{example}

\section{Conclusion} 
In this paper, we give some answers to open questions raised in \cite{CCMacta}.
First, we formalize an explicit definition of the language associated with a multi-tilde-bar 
expression, which allows us to give a recursive computation of its Glushkov functions. 
Next, we show that the worst case time complexity to construct the position automaton depends on the worst case time complexity of the function $\nul(\E)$. This function
can straightforwardly be replaced by another type of function in order to control the application of each tilde or bar. 
Last, we provide an algorithm to convert a multi-tilde-bar expression into its position automaton, with a cubic worst case time complexity  with respect to the size of the multi-tilde-bar expression.


\end{document}